\documentclass[reqno]{amsart}

\usepackage{enumitem}
\usepackage{hyperref}
\usepackage{amsthm, amssymb, amsmath,mathrsfs}
\usepackage[utf8]{inputenc}

\newtheorem{Th}{Theorem}[section]
\newtheorem{Cor}{Corollary}[section]
\newtheorem{Lemma}{Lemma}[section]
\newtheorem{Def}{Definition}[section]
\newtheorem{Rem}{Remark}[section]

\newcommand{\E}{\mathrm{e}}
\newcommand{\I}{\mathrm{i}}
\newcommand{\floor}[1]{\lfloor#1 \rfloor}

\newcommand*{\mailto}[1]{\href{mailto:#1}{\nolinkurl{#1}}}
\newcommand{\arxiv}[1]{\href{http://arxiv.org/abs/#1}{arXiv:#1}}

\newcommand{\msc}[1]{\href{http://www.ams.org/msc/msc2010.html?t=&s=#1}{#1}}

\numberwithin{equation}{section}

\begin{document}

\title{A Dynamic Uncertainty Principle for Jacobi Operators}

\author[I. Alvarez-Romero]{Isaac Alvarez-Romero}
\address{Department of Mathematical Sciences,
Norwegian University of
Science and Technology, NO--7491 Trondheim, Norway}
\email{\mailto{isaac.romero@math.ntnu.no}\\
\mailto{isaacalrom@gmail.com}}

\author[G. Teschl]{Gerald Teschl}
\address{Faculty of Mathematics\\ University of Vienna\\
Oskar-Morgenstern-Platz 1\\ 1090 Wien\\ Austria\\ and International 
Erwin Schr\"odinger Institute for Mathematical Physics\\ 
Boltzmanngasse 9\\ 1090 Wien\\ Austria}
\email{\mailto{Gerald.Teschl@univie.ac.at}}
\urladdr{\url{http://www.mat.univie.ac.at/~gerald/}}

%\date{}
\thanks{{\it Research supported by the Norwegian Research Council project DIMMA  213638.}}
\thanks{J. Math. Anal. Appl. {\bf 449}, 580--588 (2017)}


\begin{abstract}
We prove that  a solution of the Schr\"odinger-type equation $\mathrm{i}\partial_t u= Hu$, where $H$ is a Jacobi operator with asymptotically constant coefficients,
cannot decay too fast at two different times unless it is trivial.
\end{abstract}

\maketitle

\section{Introduction}

The Hardy Uncertainty Principle has been studied by several authors in the continuous case, see for example the monograph \cite{HJ} or the recent articles \cite{CEKPV,EKPV} and the references therein.
The dynamic version for the free Schr\"odinger equation says that if $u(t,x)$ is a solution of $\partial_t u=\I\Delta u$ and $|u(0,x)|=O(\E^{-x^2/\beta^2})$, $|u(1,x)|=O(\E^{-x^2/\alpha^2})$, with $1/\alpha\beta>1/4$,
then $u\equiv 0$ and if $1/\alpha\beta=1/4$, then the initial data is a constant multiple of $\E^{-(1/\beta^2+\I/4)x^2}$.

Similar results for the discrete Schr\"odinger equation have been obtained recently \cite{FB,FB2,FB3,FBV,JLMP}.
In particular, our present paper is motivated by the following result from Jaming, Lyubarskii,  Malinnikova, and Perfekt \cite{JLMP}
for the discrete Laplacian, that is $\Delta_d f(n)=f(n-1)-2f(n)+f(n+1)$:

\begin{Th}[\cite{JLMP}]\label{Ith0}
Let $u(t,n)\in C^1(\mathbb{R},\ell^2(\mathbb{Z}))$ be a solution of
\begin{equation}\label{Ischeq0}
\I \partial_t u(t,n)=\Delta_d u(t,n)+V(n)u(t,n),\qquad n\in\mathbb{Z},\quad t\in[0,1],
\end{equation}
where the potential $V(n)$ is real-valued and compactly supported
(i.e.\ $V(n)\neq 0$ only for a finite number of $n$'s).
If for some  $\epsilon>0$, 
\begin{equation*}
|u(t,n)|<C\Big(\frac{\E}{(2+\epsilon)n}\Big)^n,\qquad t\in\{0,1\}, \quad n>0,
\end{equation*}
then $u\equiv 0$.
\end{Th}

Moreover, in \cite{JLMP} the question was raised to extend this result to the case of potentials with fast
decay, not necessarily compactly supported. It is the main purpose of the present paper to provide such an extension.
In fact, we will also be slightly more general and treat Jacobi operators
\begin{equation}
H f(n)=a(n)f(n+1)+a(n-1)f(n-1)+b(n)f(n)
\end{equation}
in the Hilbert space of square summable sequences $\ell^2(\mathbb{Z})$.

\begin{Th}\label{Ith1}
Let $u(t,n)\in C^1(\mathbb{R},\ell^2(\mathbb{Z}))$ be a solution of
\begin{equation}\label{Ischeq}
\I \partial_t u= H u.
\end{equation}
 Suppose that the sequences $a(n),b(n)$, which define the Jacobi operator $H$, fulfill 
\begin{itemize}
\item[(i)] $a,b\in \ell^\infty(\mathbb{Z},\mathbb{R}),$ $a(n)>0$ and $n(1-2a(n))$, $nb(n)\in \ell^1(\mathbb{Z})$
\item[(ii)] $\sum_{n\geq N}\big(|2a(n)-1|+|b(n)|\big)\leq C\frac{1}{N^{(1+\delta)2N}}$ for $N>0$, where $C,\delta>0$ are some given constants.
\end{itemize}
If for some $\epsilon>0$, $C>0$,
\begin{equation}\label{3uconv}
|u(t,n)|\leq C\Big(\frac{\E}{(4+\epsilon)n}\Big)^n,\qquad n>0,\qquad t\in\{0,1\},
\end{equation}
then $u\equiv 0$.
\end{Th}

\begin{Rem}
\begin{enumerate}
\item Condition (i) is used to assure the existence of the Jost solutions for the Jacobi operator associated to \eqref{Ischeq}.
Condition (ii) is used to ensure an analytic extension of one reflection coefficient to the
interior of the punctured unit disk.
\item
The case where $(a(n),b(n))$ approach limits different from $(\frac{1}{2},0)$ can be easily reduced to this case using that $v(t,n)=u(\alpha t,n) \E^{-\I \beta t}$ solves
$\I\partial_t v= (\alpha H +\beta)v$.
\item
In the case of two arbitrary times $t_0<t_1$ the condition reads
\begin{equation*}
|u(t,n)|\leq C\left(\frac{(t_1-t_0)\E}{(4+\epsilon)n}\right)^n,\qquad n>0,\qquad t\in\{t_0,t_1\}.
\end{equation*}
\item
By reflecting the coefficients $\tilde{a}(n)=a(-n-1)$, $\tilde{b}(n)=b(-n)$ such that $\tilde{u}(t,n)=u(t,-n)$ solves $\I\partial_t \tilde{u}=\tilde{H} \tilde{u}$
we get a corresponding result on the negative half line.
\item
Let $w(n)\ge 1$ be some weight with $\sup_n( |\frac{w(n+1)}{w(n)}| + |\frac{w(n)}{w(n+1)}| )<\infty$ and
fix some $1 \le p \le \infty$. Set
\[
\|u\|_{w,p} = \begin{cases}
\left(\sum\limits_{n \in \mathbb{Z}} w(n) |u(n)|^p \right)^{1/p}, & 1\le p <\infty\\
\sup\limits_{n \in \mathbb{Z}} w(n) |u(n)|, & p=\infty.
\end{cases}
\]
Then one can solve \eqref{Ischeq} in the corresponding space $\ell^{w,p}(\mathbb{Z})$ and get a unique global solution
in these Banach spaces (note that our assumption ensures that the shift operators are continuous with respect to these norms).
This shows that certain decay rates (up to exponential type) are preserved by the time evolution.
\end{enumerate}
\end{Rem}

To prove this theorem we follow a similar strategy as in \cite{JLMP} using growth of entire functions and scattering theory of Jacobi operators.
It will be given in Section~\ref{secproof}.

We also mention another simple unique continuation type result inspired by \cite{KT}.

\begin{Th}
Let $u(t,n),v(t,n)\in \ell^2(C^1[0,1],\mathbb{Z})$ be strong solutions of
\begin{equation}
\I \partial_t u= H u.
\end{equation}
Suppose $a,b\in \ell^\infty(\mathbb{Z},\mathbb{R})$, $a(n)>0$. Given $n_0\in\mathbb{Z}$ and $t_0<t_1$
\begin{equation}
u(t,n)=v(t,n) \qquad\text{for}\quad n\in\{n_0,n_0+1\},\: t\in(t_0,t_1)
\end{equation}
implies $u\equiv v$.
\end{Th}

\begin{proof}
Consider $w(t,n)=u(t,n)-v(t,n)$. Then plugging the assumption $w(t,n)=0$ for $n=n_0,n_0+1$, $t\in(t_0,t_1)$ into the
differential equation implies $w(t,n)=0$ for $n=n_0-1$, $t\in(t_0,t_1)$ as well as for $n=n_0+2$, $t\in(t_0,t_1)$.
Hence the claim follows by applying this argument recursively.
\end{proof}

%%%%%%%
%%%%%%%
%%%%%%%

\section{Preliminaries}

In this section we are going to collect some results on the growth of entire functions, all of which can be found in \cite{Levin}, especially in lectures 1 and 8.
We will also give a brief introduction to Jacobi operators and their Jost solutions which can be found in Chapter~10 of \cite{Tes}.

\subsection{Growth of entire functions}

Let $f(z)$ be an entire function. We say that $f$ is of \textbf{exponential type} $\sigma_f$ if for $|z|$ big enough and some $\sigma>0$ we always have
\begin{equation}\label{Gdecay}
|f(z)|<\exp(\sigma |z|).
\end{equation}
The type $\sigma_f$ of the function $f$ is defined by
\begin{equation*}
\sigma_f= \limsup_{r\to\infty} \frac{\log \max\{|f(r \E^{\I\varphi})|:\varphi\in[0,2\pi]\}}{r}
\end{equation*}

\begin{Th}\label{21cauchy}
Let $f(z)=\sum_{n\geq 0}c_nz^n$, be an entire function, then the type of $f$ can be determined via the formula
\begin{equation}
\limsup_{n\to\infty} n|c_n|^{1/n}=\E\, \sigma_f
\end{equation}
\end{Th}

So far we have considered the growth of $f(z)$ in all directions simultaneously, but it may happen that the function behaves different along different directions. To this end we introduce the indicator function
\begin{equation}\label{Gindi}
h_f(\varphi)=\limsup_{r\to\infty}\frac{\log|f(r \E^{\I\varphi})|}{r},
\end{equation}
where $\varphi$ denotes the direction we are interested in, i.e.\ $\arg(z)=\varphi$.

It follows from the definition that
\begin{equation}\label{Gsuma}
h_{f+g}\leq\max(h_f,h_g)
\end{equation}
and
\begin{equation}\label{Gprod}
h_{fg}\leq h_f+h_g.
\end{equation}

\begin{Def}
A function $K(\theta)$ is called trigonometrically convex on the closed segment  $[\alpha,\beta]$ if for $\alpha\leq \theta_1<\theta_2\leq\beta$, $0<\theta_2-\theta_1<\pi$ we have
\[
K(\theta)\leq \frac{K(\theta_1)\sin(\theta_2-\theta)+K(\theta_2) \sin(\theta-\theta_1)}{\sin(\theta_2-\theta_1)},\qquad \theta_1\leq \theta\leq \theta_2.
\]
\end{Def}

\begin{Th}\label{Gindicator}
Let $f(z)$ be an entire function of exponential type. Then its indicator function $h_f$ is a trigonometrically convex function.
\end{Th}

As a consequence we note

\begin{Cor}
Let $f(z)$ be an entire function of exponential type, then
\begin{equation}\label{Gindic}
h_f(\varphi)+h_f(\pi+\varphi)\geq 0.
\end{equation}
\end{Cor}

\begin{Rem}
The key part of the proof of the Theorem~\ref{Gindicator} is the Phragm\'en--Lindel\"of theorem, thus one can easily adapt the proof of Theorem~1 from Chapter~8 in \cite{Levin}
to show that it continuous to hold if $f$ is only analytic in a region $\{z: |z| > \rho\}$. In particular, inequality \eqref{Gindic} is still true in this case.
\end{Rem}

%%%%%%

\subsection{Jacobi operators and Jost solutions}

Suppose 
\begin{equation*}
a,b\in \ell^\infty(\mathbb{Z},\mathbb{R}),\qquad a(n)> 0
\end{equation*}
and consider the associated self-adjoint Jacobi operator 
\begin{equation*}
\begin{split}
H:\quad\ell^2(\mathbb{Z})&\to \ell^2(\mathbb{Z}),\\
f&\mapsto \tau f,
\end{split}
\end{equation*}
where
\begin{equation*}
\tau f(n)=a(n)f(n+1)+a(n-1)f(n-1)+b(n)f(n).
\end{equation*}
In fact, we will make the stronger assumption
\begin{equation}\label{Jhyp2}
n(2a(n)-1)\in \ell^1(\mathbb{Z}),\quad nb(n)\in \ell^1(\mathbb{Z}).
\end{equation}
We recall \cite{Tes} that under this assumption the spectrum of $H$ consists of an purely absolutely continuous part
covering $[-1,1]$ plus a finite number of discrete eigenvalues in $\mathbb{R}\setminus[-1,1]$. The associated
spectral equation is
\begin{equation}\label{Jspec}
\tau f = \lambda f
\end{equation}
where $\lambda $ is a complex number and there are two independent solutions. The Wronskian of two solutions is given by
\begin{equation*}
W(f,g)=a(n)\big(f(n)g(n+1)-g(n)f(n+1)\big)
\end{equation*}
and does not depend on $n$ if $f,g$ both solve \eqref{Jspec}. Instead of $\lambda$ it is more convenient to use
$\theta\in\mathbb{T}:=\{z: |z|=1\}$ given by
\begin{align*}
\begin{split}
\lambda:\quad &\mathbb{T}\to [-1,1],\\
&\theta\mapsto \lambda(\theta):=\frac{1}{2}(\theta+\theta^{-1}).
\end{split}
\end{align*}

\begin{Th}
Let $a(n),b(n)$ be as in \eqref{Jhyp2}, then there exists solutions to \eqref{Jspec}, called Jost solutions, $e^\pm(\theta,n)$, $0<|\theta|\leq 1$, fulfilling
\begin{equation*}
\lim_{n\to\pm\infty}e^{\pm}(\theta,n)\theta^{\mp n}=1,\quad 0<|\theta|\leq 1.
\end{equation*}
\end{Th}

We can write the Jost solutions in terms of Fourier series via
\begin{equation}\label{22fourier}
\begin{split}
e^{+}(\theta,n)&=\frac{\theta^n}{A_{+}(n)}\big(1+\sum_{j=1}^\infty K_{+,j}(n)\theta^j\big),\text{  }|\theta|\leq 1,\\
e^{-}(\theta,n)&=\frac{\theta^{-n}}{A_{-}(n)}\big(1+\sum_{j=1}^\infty K_{-,j}(n)\theta^j\big),\quad |\theta|\leq 1,
\end{split}
\end{equation}
where $A_{-}(n)=\prod_{m=-\infty}^{n-1}2a(m)$, and $A_{+}(n)=\prod_{m=n}^\infty 2a(m)$. Notice that $A_{\pm}(n)$ are uniformly bounded due to \eqref{Jhyp2}.
For later use we will also set $K_{\pm,0}(n):=1$.

Moreover, the coefficients $K_{+,j}(n)$ are bounded by
\begin{equation}\label{22bddcoe}
|K_{+,j}(n)|\leq D_{+,j}(n)C_{+}(n+\floor{\frac{j}{2}}+1),\qquad j\in\mathbb{N},
\end{equation}
where
\begin{equation}\label{22bddcmas}
C_{+}(n)=\sum_{m=n}^\infty c(m),\quad
D_{+,m}(n)=\prod_{j=1}^{m-1}(1+C_{+}(n+j)), \quad
c(n)=2|b(n)|+|4a(n)^2-1|
\end{equation}
and $\floor{x}=\max \{ n \in \mathbb{Z} | n \leq x\}$ is the usual floor function.
Notice that $\{D_{+,m}(n)\}_{m,n\in\mathbb{N}}$ is a bounded set. For $K_{-,j}(n)$ we have analogous results.

We already know that the Wronskian does not depend on $n$, whence we observe that the Jost solutions $e^{\pm}(\theta,n),e^{\pm}(\theta^{-1},n)$ are independent for $|\theta|=1$, $\theta^2\neq 1$:

\begin{equation*}
W(e^{\pm}(\theta),e^{\pm}(\theta^{-1}))=\pm\frac{1-\theta^2}{2\theta}.
\end{equation*}
Moreover, they can be expressed as
\begin{equation}\label{scatrel}
e^{\pm}(\theta,n)=\alpha(\theta)e^{\mp}(\theta^{-1},n)+\beta_{\mp}(\theta)e^{\mp}(\theta,n),\quad |\theta|=1,
\end{equation}
where
\begin{equation}\label{Jscat}
\begin{split}
\alpha(\theta)&=\frac{W(e^{\mp}(\theta),e^{\pm}(\theta))}{W(e^{\mp}(\theta),e^{\mp}(\theta^{-1}))}=\frac{2\theta}{1-\theta^2}W(e^{+}(\theta),e^{-}(\theta))\\
\beta_{\pm}(\theta)&=\frac{W(e^{\mp}(\theta),e^{\pm}(\theta^{-1}))}{W(e^{\pm}(\theta),e^{\pm}(\theta^{-1}))}=\pm\frac{2\theta}{1-\theta^2}W(e^{\mp}(\theta),e^{\pm}(\theta^{-1}))
\end{split}
\end{equation}
Our assumption \eqref{Jhyp2} implies $c \in \ell^1(\mathbb{Z})$ and hence $e^{\pm}(.,n)$ are analytic inside the unit disc $\mathbb{D} := \{z : |z|<1\}$ and continuous up to the boundary. Consequently
$\alpha$ is analytic inside the unit disc
\begin{equation}\label{alpha20}
\alpha(\theta)=\frac{1}{A}\sum_{j\geq 0}K_{j}\theta^j, \qquad A=\prod_{m=-\infty}^\infty2a(m)>0,
\end{equation}
with $K_j=\lim_{n\to\mp\infty}K_{\pm,j}$ (in particular $K_0=1$). The only zeros inside $\mathbb{D}$ of $\alpha$ are the eigenvalues and hence there are only finitely many. For later use we record
the trivial consequence
\begin{equation}\label{3amas}
\limsup_{|\theta|\to\infty}\frac{\log|\alpha(\theta^{-1})|}{|\theta|}= 0.
\end{equation}
Moreover, the additional assumption
\begin{equation}\label{Sdec0}
\sum_{n\geq N}\big(2|b(n)|+|4a(n)^2-1|\big)\leq \frac{C}{N^{(1+\delta)2N}},\quad N>0,
\end{equation}
implies 

\begin{Lemma}\label{lembeta}
Under the assumptions \eqref{Jhyp2} and \eqref{Sdec0} we have that $e^+(.,n)$ is an entire function satisfying
\begin{equation}
\limsup_{|\theta|\to\infty} \frac{\log|e^+(\theta,n)|}{|\theta|}\leq 0.
\end{equation}
\end{Lemma}

\begin{proof}
This is a simple application of Theorem~\ref{21cauchy} using \eqref{Sdec0} and \eqref{22bddcoe}.
\end{proof}

As a consequence we note that $\beta_+$ is analytic in the punctured unit disc $\mathbb{D}\setminus\{0\}$ and satisfies
\begin{equation}\label{3bmas}
\limsup_{|\theta|\to\infty} \frac{\log|\beta_+(\theta^{-1},n)|}{|\theta|}\leq 0.
\end{equation}

%%%%%%%%%%%%

\section{Schr\"odinger evolutions}
\label{secproof}

Consider
\begin{equation}
\mathcal{F}(f)(\theta) = \sum_{n\in\mathbb{Z}} f(n) \begin{pmatrix} e^+(\theta,n)\\ e^-(\theta,n)\end{pmatrix},
\end{equation}
then $\mathcal{F}: \ell^2(\mathbb{Z}) \to L^2(\mathbb{T}_+ \cup \{\theta_j\},d\rho)$ is unitary such that
\begin{equation}\label{sptrans}
\mathcal{F}(H f)(\theta) = \lambda(\theta) \mathcal{F}(f)(\theta),
\end{equation}
where
\begin{equation}
d\rho(\theta) = \begin{pmatrix} 1&0\\0&1\end{pmatrix} \frac{d\theta}{2\pi\I \theta |\alpha(\theta)|^2} +\sum_{j=1}^k \begin{pmatrix} \gamma_j&0\\0&0\end{pmatrix} d\Theta(\theta-\theta_j)
\end{equation}
is the associated spectral measure. Here $\theta_j$ are the eigenvalues of $H$, $\gamma_j^{-1} := \sum_{n\in\mathbb{Z}} |e^+(\theta_j,n)|^2$ are the corresponding norming constants,
and $d\Theta(\theta-\theta_j)$ is a Dirac measure centered at $\theta_j$.

In particular, if
\begin{equation*}
u(t)= \E^{-\I t H} u(0)
\end{equation*}
is the solution of \eqref{Ischeq0}, then
\begin{equation*}
\mathcal{F}(u(t))(\theta) = \E^{-\I\lambda(\theta)} \mathcal{F}(u(0))(\theta), \qquad \lambda(\theta)=\frac{1}{2}(\theta^{-1}+\theta).
\end{equation*}

\begin{proof}[Proof of Theorem~\ref{Ith1}.]
Consider the auxiliarly function $\Phi(t,\theta)$ defined as (using \eqref{scatrel})
\begin{align}\nonumber
\Phi(t,\theta) :&\!=\sum_{n\in\mathbb{Z}}u(t,n) e^-(\theta,n)\\ \label{3phiAB}
& \!= \sum_{n<0}u(t,n)e^-(\theta,n)+\beta_+(\theta)\sum_{n\geq 0}u(t,n)e^+(\theta,n)+\alpha(\theta) \sum_{n\geq 0}u(t,n)e^+(\theta^{-1},n)\\ \nonumber
&\! =: A_1(t,\theta)+\beta_+(\theta)A_2(t,\theta)+\alpha(\theta)B(t,\theta).
\end{align}
Due to our assumption \eqref{Jhyp2} and the estimate \eqref{22bddcoe} the two sums $A_k(t,\theta)$ converge compactly with respect to $\theta\in\mathbb{D}$ and hence represent
analytic functions on $\mathbb{D}$. Moreover, by Lemma~\ref{lembeta} $A(t,\theta):= A_1(t,\theta)+\beta_+(\theta)A_2(t,\theta)$ is analytic in $\mathbb{D}\setminus\{0\}$ and satisfies
\[
\limsup_{|\theta|\to\infty} \frac{\log|A(t,\theta^{-1})|}{|\theta|}\leq 0.
\]
By \eqref{3amas} it remains to study
\[
B(t,\theta):=B_1(t,\theta^{-1})+B_2(t,\theta^{-1}),
\]
where $B_1(t,\theta):=\sum_{n\geq 0}v(t,n) \theta^n$ and $B_2(t,\theta):=\sum_{n\geq 0}v(t,n)\sum_{j\geq 1}K_{+,j}(n)\theta^{j+n}$
and $v(t,n):=\frac{u(t,n)}{A_+(n)}$. Note that $v(t,.)\in\ell^2(\mathbb{Z})$ also satisfies \eqref{3uconv} (but of course with a different constant in general)
and hence $B_1(t,.)$ is entire with
\begin{equation}\label{limB1}
\limsup_{|\theta|\to\infty}\frac{\log|B_1(t,\theta)|}{|\theta|}\leq \frac{1}{4+\epsilon},\quad t\in\{0,1\}.
\end{equation}

Due to \eqref{22fourier} and \eqref{22bddcoe}, the series $B_2(t,\theta)$ is absolutely convergent for $t\in\{0,1\}$ and  we have
\begin{equation*}
B_2(t,\theta)=\sum_{n\geq 0} v(t,n) \sum_{j\geq 1}K_{+,j}(n)\theta^{j+n}=\sum_{j=1}^\infty b_j(t)\theta^j,\quad t\in\{0,1\},
\end{equation*}
where $b_j(t):=\sum_{k=0}^{j-1} v(t,k)K_{+,j-k}(k)$. Moreover, by \eqref{22bddcoe}, \eqref{Sdec0}
\begin{align*}
|K_{+,j}(n)| &\leq D_{+,j}(n)C_{+}(n+\floor{\frac{j}{2}}+1)\leq \frac{C}{(n+\floor{j/2}+1)^{(1+\delta)2(n+\floor{j/2}+1)}}\\
&\leq C \left(\frac{2}{j+n}\right)^{(1+\delta)(j+n)}.
\end{align*}
Whence using \eqref{3uconv}
\begin{equation*}
|b_j(t)| \leq \sum_{k=0}^{j-1} |v(t,k)| |K_{+,j-k}(k)| \leq
C \left(\frac{2}{j}\right)^{(1+\delta)j} \sum_{k=0}^{j-1} \Big(\frac{e}{(4+\epsilon)k}\Big)^k
\end{equation*}
here $C>0$ is a constant. Thus $B_2(t,.)$ is entire with
\begin{equation*}
\limsup_{|\theta|\to\infty}\frac{\log|B_2(t,\theta)|}{|\theta|}\leq0,\quad t\in\{0,1\}.
\end{equation*}
In summary we have
\begin{equation*}
\limsup_{|\theta|\to\infty}\frac{\log|B(t,\theta^{-1})|}{|\theta|}\leq \limsup_{|\theta|\to\infty}\frac{\log|B_1(t,\theta)|}{|\theta|}\leq \frac{1}{4+\epsilon},\quad t\in\{0,1\}
\end{equation*}
and therefore
\begin{equation*}
\limsup_{|\theta|\to\infty}\frac{\log|\Phi(t,\theta^{-1})|}{|\theta|}\leq\frac{1}{4+\epsilon},\quad t\in\{0,1\}.
\end{equation*}
Using inequality \eqref{Gindic}, this implies
\begin{equation*}
\begin{split}
0&\leq \limsup_{r\to\infty}\frac{\log|\Phi(t,r^{-1} \E^{\I\pi/2})|}{r}+ \limsup_{r\to\infty}\frac{\log|\Phi(t,r^{-1} \E^{-\I\pi/2})|}{r}\\
&\leq \frac{1}{4+\epsilon}+\limsup_{r\to\infty}\frac{\log|\Phi(t,r^{-1} \E^{\pm \I\pi/2})|}{r},\quad t\in\{0,1\},
\end{split}
\end{equation*}
that is,
\begin{equation*}
\begin{split}
\limsup_{r\to\infty}\frac{\log|\Phi(t,r^{-1} \E^{\pm \I\pi/2})|}{r}\geq-\frac{1}{4+\epsilon},\quad t\in\{0,1\}.
\end{split}
\end{equation*}
On the other hand, by \eqref{sptrans} we have
\begin{equation*}
\Phi(t,\theta)=\E^{-\I t\lambda(\theta)}\Phi(0,\theta)
\end{equation*}
for $|\theta|=1$ in the sense of $L^2$. Since we have seen that $\Phi(1,\theta)$ is analytic for $\theta\in\mathbb{D}\setminus\{0\}$ and continuous
up to $\mathbb{T}$ we conclude that
\begin{equation}\label{Phi}
\Phi(1,\theta)= \E^{-\I \lambda(\theta)}\Phi(0,\theta), \qquad 0<|\theta|\le 1.
\end{equation}
But this is not possible unless $\Phi\equiv 0$, since by \eqref{Phi} we have
\begin{equation*}
\limsup_{y\to\infty}\frac{\log|\Phi(1,-\I y^{-1})|}{y}= \frac{1}{2} +\limsup_{y\to\infty}\frac{\log|\Phi(0,-\I y^{-1})|}{y}\geq \frac{1}{2}-\frac{1}{4+\epsilon}
>\frac{1}{4+\epsilon}
\end{equation*}
In addition, we know 
\[
\limsup_{y\to\infty}\frac{\log|B_1(1,\I y)|}{y}\geq\limsup_{y\to\infty}\frac{\log|B(1,-\I y^{-1})|}{y}\geq\limsup_{y\to\infty}\frac{\log|\Phi(1,-\I y^{-1})|}{y},
\]
contradicting \eqref{limB1} unless $\Phi(1,\theta)\equiv0$.
But this implies $\mathcal{F}_1(u(1))(\theta)=0$ for $\theta\in\mathbb{T}$ and $\theta=\theta_j$.
Using \eqref{scatrel} we also get $\mathcal{F}_2(u(1))(\theta)=0$ for $\theta\in\mathbb{T}$ and hence
$\mathcal{F}(u(1))\equiv 0$, that is $u(t)\equiv0$. 
\end{proof}

%%%%%%%
\medskip

%%%%%%%%
\noindent 
{\bf Acknowledgments.}
We are indebted to Yura Lyubarskii for discussions on this topic and to the anonymous referee for valuable remarks leading to an improved presentation.
I.\ A-R.\ gratefully acknowledges the hospitality of the Faculty of Mathematics, University of Vienna, Austria, during May, June 2016 where this research was performed.
%%%%%%%%

\end{document}